\documentclass[a4paper,12pt,twoside]{amsart}
\usepackage[left=2.5cm, right=2.5cm]{geometry}
\usepackage{type1cm}
\usepackage[utf8]{inputenc}
\usepackage[T1]{fontenc}
\usepackage{aecompl}
\usepackage{ae}
\usepackage{float}
\usepackage{textcomp}
\usepackage[english]{babel}
\usepackage[usenames]{color}

\usepackage{hyperref}

\usepackage{amsmath,amsthm, amsfonts}
\usepackage{bbm}
\usepackage{enumerate}

\theoremstyle{plain}
\newtheorem{thm}{Theorem}[section]

\newtheorem{cor}[thm]{Corollary}
\newtheorem{lem}[thm]{Lemma}

\theoremstyle{definition}

\newcommand{\entr}{\mathsf{h}}
\newcommand{\eps}{\varepsilon}

\newcommand{\N}{\mathbb{N}}
\newcommand{\R}{\mathbb{R}}
\newcommand\Si{\mathbf{\Sigma}}

\newcommand{\abs}[1]{\lvert#1\rvert}

\hypersetup{%
  pdftitle = {Growth-sensitity},
  pdfauthor = {Wilfried Huss, Ecaterina Sava, Wolfgang Woess},
  pdfkeywords = {Languages, Random Walks on graphs, Spectral radius}
}
\begin{document}
\title{Entropy sensitivity of languages defined by infinite automata, via
Markov chains with forbidden transitions}
\author{Wilfried Huss, Ecaterina Sava, Wolfgang Woess}
\address{\parbox{.8\linewidth}{Institut f\"ur Mathematische Strukturtheorie\\ 
Technische Universit\"at Graz\\
Steyrergasse 30, 8010 Graz, Austria\\}}
\email{huss@finanz.math.tu-graz.ac.at, sava@TUGraz.at, woess@TUGraz.at}
\date{\today} 
\thanks{W.~Huss and W.~Woess were supported by the
Austrian Science Fund project FWF-P19115-N18. E.~Sava was supported
by the NAWI Graz project.}

\begin{abstract} 
A language $L$ over a finite alphabet $\Si$ is growth-sensitive 
(or entropy sensitive) if 
forbidding any finite set of factors $F$ of $L$ yields a sub-language $L^F$ whose 
exponential growth rate (entropy) is smaller than that of $L$. Let 
$(X, E, \ell)$ be an infinite, oriented, edge-labelled graph with label 
alphabet $\Si$. Considering the graph as an (infinite) automaton,  
we associate with any pair of vertices $x,y \in X$ the language $L_{x,y}$ 
consisting of all words that can be read as the labels along some path 
from $x$ to $y$. Under suitable, general assumptions we prove that these 
languages are growth-sensitive. This is based on using Markov chains with 
forbidden transitions. 
\end{abstract}

\maketitle
\markboth{{\sf W. Huss, E. Sava, W. Woess}}
{{\sf Entropy sensitivity}}
\section{Introduction}\label{sec:intro}

Let $\Si$ be a finite alphabet and $\Si^{*}$ the set of all finite 
words over $\Si$, including the empty word $\epsilon$. 
A language $L$ over $\Si$ is a subset of $\Si^{*}$.
All our languages will be infinite.  
We denote by $\abs{w}$ the length of the word $w$. 
A \emph{factor} of a word $w=a_1a_2\ldots a_n$ is
a word of the form $a_ia_{i+1}\ldots a_j$, with $1\leq i\leq j\leq n$. 
The \emph{growth} or \emph{entropy} of $L$ is 
\begin{equation*}
\entr(L) 
= \limsup_{n\to\infty}
\frac{1}{n} \log \bigl|\{w \in L:\: \abs{w} = n\}\bigr|.
\end{equation*}
For a finite, non-empty set 
$F\subset\Si^+ = \Si^*\setminus \{\epsilon\}$ consisting of factors 
of elements of $L$, we let 
\begin{equation*}
 L^{F} 
 = \{w\in L:\:\text{no}\; v\in F\; \text{is a factor of}\; w\}.
\end{equation*}
The issue addressed here is to provide conditions under which, for a
class of languages associated with infinite graphs,
$\entr(L^{F})<\entr(L)$. If this holds for \emph{any}
set $F$ of \emph{forbidden factors,} then the language $L$ is called
\emph{growth sensitive} (or \emph{entropy sensitive}). 

Questions related with growth sensitivity have been considered in different
context.

In \emph{group theory,} in relation with regular normal forms of finitely 
generated groups, the study of growth-sensitivity has been proposed by 
{\sc Grigorchuk and de la Harpe}~\cite{GrHa} as a tool for proving 
Hopfianity of a given group or class of groups, see also
{\sc Arzhantseva and Lysenok}~\cite{ArLy} and
{\sc Ceccherini-Silberstein and Scarabotti} \cite{CeSc}. 

In \emph{symbolic dynamics,} the number $\entr(L)$ associated with a regular
language accepted by a finite automaton with suitable properties
appears as the \emph{topological entropy} of a \emph{sofic system}, see
{\sc Lind and Marcus}~\cite[Chapters 3 \& 4]{LiMa}. 
Entropy sensitivity appears as the strict inequality between the entropies
of an irreducible sofic shift and a proper subshift~\cite[Cor. 4.4.9]{LiMa}.

Motivated by these bodies of work, {\sc Ceccherini-Silberstein and 
Woess}~\cite{CeWo1}, \cite{CeWo2}, \cite{Ce} have elaborated
practicable criteria that guarantee growth-sensitivity of \emph{context-free}
languages. 

The main result of the present note can be seen as a direct extension of 
\cite[Cor. 4.4.9]{LiMa} to the entropies of infinite sofic systems; see below 
for further comments and references.

\smallskip

Our basic object is an infinite oriented
graph $(X,E, \ell)$ 
whose edges are labelled by elements of a finite alphabet $\Si$. Each edge 
has the form $e=(x,a,y)$, where $e^-=x$ and $e^+= y \in X$ are the initial 
and the terminal vertex of $e$ and $\ell(e) = a \in \Si$ is its label. 
We will also write $x \xrightarrow{a} y$ for the edge $e=(x,a,y)$, or just 
$x\rightarrow y$ in situations where we do not care about the label.
Multiple edges and loops are allowed, but two edges with the same end 
vertices must have distinct labels. 

A \emph{path} of length $n$ in $(X,E, \ell)$ is a sequence $\pi=e_{1}e_{2}\ldots e_{n}$ 
of edges such that $e_{i}^{+}=e_{i+1}^{-}$, for $i=1,2,\ldots n-1$. 
We say that it is a path from $x$ to $y$, if $e_1^-=x$ and $e_n^+ =y$.
The label $l(\pi)$ of $\pi$ is the word
$\ell(\pi)=\ell(e_{1})\ell(e_{2})\ldots \ell(e_{n})\in \Si^{*}$ that we read 
along the path. We also allow the empty path from $x$ to $x$, whose label is 
the \emph{empty word} $\epsilon \in \Si^*$. 
For $x,y\in X$, denote by $\Pi_{x,y}$ 
the set of all paths $\pi$ from $x$ to $y$ in $(X,E, \ell)$. 

The languages which we consider here are
\begin{equation*} 
L_{x,y}
=\{\ell(\pi)\in\Si^{*}: \pi\in \Pi_{x,y}\}, 
\text{ where } x,y\in X. 
\end{equation*}
That is, we can interpret the edge-labelled graph $(X,E, \ell)$ as an infinite automaton (labelled 
digraph) with initial
state $x$ and terminal state $y$, so that $L_{x,y}$ is the language
accepted by the automaton. 

We say that $(X,E,\ell)$ is \emph{deterministic,} if for every vertex
$x$ and every $a \in \Si$, there is at most one edge with initial point $x$
and label $a$. Any automaton (finite or infinite) can be transformed into a 
deterministic one that accepts the same language, by the well known powerset construction.
See e.g. \cite[Prop. 1.4.1]{codes_automata}.

As in the finite case, we need an irreducibility assumption. 
The graph $(X,E, \ell)$ is called \emph{strongly connected,} if for every
pair of vertices $x$, $y$, there is an (oriented) path from $x$ to $y$.  
Furthermore, we say that it is \emph{uniformly connected,} if in addition the
following holds. 
\begin{itemize}
\item
There is a constant $K$ such that for very edge $x\rightarrow y$ there is
a path from $y$ to $x$ with length at most $K$.
\end{itemize}
In the finite case, the two notions coincide as one can take $K=|X|$. 
The \emph{forward distance}
$d^+(x,y)$ of $x,y \in X$ is the minimum length of a path from $x$ to $y$.
We write 
$$
\entr(X) = \entr(X,E,\ell) = 
\sup_{x,y\in X} \entr(L_{x,y})
$$
and call this the entropy of our oriented, labelled graph.
It is a well known and easy to prove fact that for a strongly connected graph,
$\entr(L_{x,y}) = \entr(X)$ for all $x, y \in X$.

We also need a reasonable assumption on the set of forbidden factors.

We say that a finite set $F \subset \Si^+$ is \emph{relatively dense}
in the graph $(X,E, \ell)$, if there is a constant $D$ such that for every
$x \in X$ there are $y \in X$ and $w \in F$ such that $d^+(x,y) \le D$
and there is a path starting at $y$ which has label $w$. 

Note that the assumptions of uniformly connectedness and relatively denseness
cannot be avoided, since they play an important role in the prove of
the main result. This fails withous this assumptions.

\begin{thm}\label{thm:A}
Suppose that $(X,E,\ell)$ is uniformly connected and deterministic with
label alphabet $\Si$. Let $F \subset \Si^+$ be a finite, non-empty set
which is relatively dense in $(X,E, \ell)$. Then 
$$
\sup_{x,y\in X}\entr(L_{x,y}^F) < \entr(X) \quad\text{strictly.}
$$
\end{thm}

We say that $(X,E, \ell)$ is \emph{fully deterministic,} if for every
$x \in X$ and $a \in \Si$, there is precisely one edge with initial point $x$
and label $a$. Remark that in automata theory, the classical terminalogy is deterministic and
complete, instead of fully deterministic.
Since in graph theory
a complete graph is one in which every pair a distinct vertices is connected 
by an unique edge, we shall use the notion of fully deterministic graphs
throughout this paper.

\begin{cor}\label{cor:B}
If $(X,E,\ell)$ is uniformly connected and fully deterministic then
$L_{x,y}$ is growth-sensitive for all $x,y \in X$.
\end{cor}

Indeed, in this case, for every $x \in X$ and every $w \in \Si^*$, there
is precisely one path with label $w$ starting at $x$.

With our edge-labelled graph $(X,E,\ell)$, we can consider the \emph{full shift space} 
which consists of all bi-infinite words over $\Si$ that can be read along the edges of 
some bi-infinite path in $(X,E, \ell)$. When $(X,E,\ell)$ is strongly connected,
the entropy $\entr(L_{x,y})$ is independent of $x$ and $y$ and 
equals the topological entropy of the full shift space of the graph.
See e.g. {\sc Gurevi\v c}~\cite{Gu}, 
{\sc Petersen}~\cite{Pe} or {\sc Boyle, Guzzi and G\'omez}~\cite{BGG}
for a selection of related work and references, and also
the discussion in \cite[\S 13.9]{LiMa}.  

If we consider the shift space consisting of all those
bi-infinite words as above that do not contain any factor in $F$,  
then the interpretation of Corollary \ref{cor:B} is that
the associated entropy is strictly smaller than $\entr(X)$.

The theorem, once approached in the right way, is not hard to prove.
It is based on a classical tool, a version of the Perron-Frobenius 
theorem for infinite
non-negative matrices; see e.g. {\sc Seneta}~\cite{Se}. We shall first
reformulate things in terms of Markov chains 
and forbidden transitions.

\section{Markov chains and forbidden transitions}\label{sec:Markov} 
We now equip the oriented, edge-labelled graph $(X,E,\ell)$ with additional data:
with each edge $e=(x,a,y)$, we associate a probability 
$p(e) = p(x,a,y) \ge \alpha > 0$,  where $\alpha$ is a fixed constant,
such that
\begin{equation}\label{eq:substoch}
\sum_{e \in E \,:\,e^-=x} p(e) \le 1 \quad\text{for every}\; x \in X\,.
\end{equation}
Our assumption to have the uniform lower bound $p(e) \ge \alpha$ for each 
edge implies that the outdegree (number of outgoing edges) of each vertex
is bounded by $1/\alpha$.
We interpret $p(e)$ as the probability that a particle with current position
$x=e^-$ moves in one (discrete) time unit along $e$ to its end vertex $y = e^+$.
Observing the successive random positions of the particle at the time 
instants $0,1,2,\dots$, we obtain a Markov chain with state space $X$ whose
one-step transition probabilities are 
$$
p(x,y) = \sum_{a \in \Si : (x,a,y) \in E} p(x,a,y)\,.
$$
We shall also want to record the edges, resp. their labels used in each step,
which means to consider a Markov chain on a somewhat larger 
state space, but we will not need to formalise this in detail.
In \eqref{eq:substoch}, we admit the possibility that $1 - \sum_y p(x,y) > 0$
for some $x$. This number is then interpreted as the probability that a
particle positioned at $x$ dies at the next step.

We write $p^{(n)}(x,y)$ for the probability that the particle starting at
$x$ is at position $y$ after $n$ steps. This is the $(x,y)$-element of
the $n$-power $P^n$ of the transition matrix 
$P = \bigl( p(x,y) \bigr)_{x,y \in X}\,$. If $(X,E, \ell)$ is strongly connected,
then $P$ is irreducible, and it is well-known that the number
$$
\rho(P) = \limsup_{n \to \infty} p^{(n)}(x,y)^{1/n}
$$
is independent of $x$ and $y$. See once more \cite{Se}.
Often, $\rho(P)$ is called the spectral radius of $P$. It is the parameter
of exponential decay of the transition probabilities.

Let once more $F \subset \Si^+$ be finite. We interpret the elements of
$F$ as sequences of \emph{forbidden transitions.} That is, we restrict 
the motion of the particle: at no time, it is allowed to traverse any path  
$\pi$ with $\ell(\pi) \in F$ in $k$ successive steps, where $k$ is
the length of $\pi$. 
We write $p^{(n)}_F(x,y)$ for the probability that the particle starting at
$x$ is at position $y$ after $n$ steps, without having made any such sequence
of forbidden transitions. Let
$$
\rho_{x,y}(P_F) = \limsup_{n \to \infty} p_F^{(n)}(x,y)^{1/n}, \quad
x,y \in X\,.
$$
These numbers are not necessarily independent of $x$ and $y$, and they are not
the elements of the $n$-matrix power of some substochastic matrix. 

Recall that a transition matrix $Q = \big(q(x,y)\big)_{x,y\in X}$ on the state space $X$
is called \emph{substochastic} if there exists a constant
$\varepsilon > 0$, such that for all $x\in X$
\begin{equation*}
\sum_{y\in X} q(x,y) \leq 1-\varepsilon.
\end{equation*}
That is, all row sums are bounded by $1-\varepsilon$.
In order to give an upper bound for the restricted transition 
probabilities $p^{(n)}_F(x,y)$, we first show the following.  

\begin{lem}\label{lem:q_str_substoch}
Suppose that $(X,E,l)$ is strongly 
connected with label alphabet $\Si$ and equipped with transition 
probabilities $p(e) \ge \alpha > 0$, $e \in E$. Let $F \subset \Si^+$ be a 
finite, non-empty set which is relatively dense in $(X,E, \ell)$. Then there are
$k\in\mathbb{N}$ and $\eps_0>0$ such that  
$$
\sum_{y \in X} p_F^{(k)}(x,y) \le 1-\eps_0 \quad \text{for all}\; x \in X\,.
$$
In other words, the transition matrix 
$Q = \bigl(p_F^{(k)}(x,y)\bigr)_{x,y \in X}$ is strictly substochastic, with all row sums bounded by $1-\eps_0\,$.
\end{lem}
\begin{proof}
Let $R = \max_{w \in F} |w|$, and let $D \in \N$ be the constant from the
definition of relative denseness of $F$. Set $k=D+R$. For each $x \in X$,
we can find a path $\pi_1$ from $x$ to some 
$y \in X$ with length $d \le D$ and a path $\pi_2$ starting at $y$ which has label $w\in\Si^*$. Let $z$
be the endpoint of $\pi_2$, and choose any path $\pi_3$ that starts at
$z$ and has length $k - d - |w|$. (Such a path exists by strong
connectedness.) Then let $\pi$ be the path obtained by concatenating
$\pi_1\,$, $\pi_2$ and $\pi_3\,$.

The probability that the Markov chain starting at $x$ makes its first $k$
steps along the edges of $\pi$ is
\begin{equation*}
\mathbb{P}(\pi)\geq \alpha ^{k}=\eps_0 > 0. 
\end{equation*}
Hence
\begin{equation*}
\sum_{y\in X} p^{(k)}_F(x,y)\leq \sum_{y\in X} p^{(k)}(x,y) -\mathbb{P}(\pi)
\leq 1-\eps_0, 
\end{equation*}
and this upper bound holds for every $x$.
\end{proof}

The matrix $P$ acts on functions $h: X \to \R$ by 
$Ph(x) = \sum_y p(x,y) h(y)$.
Next, we state two key results due to  Pruitt \cite[Lemma 1]{pruitt_1964} and 
\cite[Corollary to Theorem 2]{pruitt_1964}, which will be used in
the proof of the main result.

\begin{lem}\label{lem:pruitt1}
If the transition matrix $P$ is irreducible and $Ph \leq s h$ for some 
$s > 0$ and $h \not= 0$, then $h > 0$.
\end{lem}

\begin{lem}\label{lem:pruitt2}
If the transition matrix $P = \{ p(x,y)\}_{x,y\in X}$ is such that 
for every $x\in X$ the entries $p(x,y) = 0$ for all $y\in X$
except finitely many, then the equation
\begin{equation*}
P h = s h
\end{equation*}
has a solution for all $s \geq \rho(P)$.
\end{lem}
Using these lemmatas, we prove the following result on sensitivity
of the Markov chain with respect to forbidding the transitions in $F$.

\begin{thm}\label{thm:C}
Suppose that $(X,E,\ell)$ is uniformly 
connected with
label alphabet $\Si$ and equipped with transition probabilities
$p(e) \ge \alpha > 0$, $e \in E$. Let $F \subset \Si^+$ be a 
finite, non-empty set which is relatively dense in $(X,E, \ell)$.
Then 
$$
\sup_{x,y \in X} \rho_{x,y}(P_F) < \rho(P) \quad\text{strictly.}
$$
\end{thm}

\begin{proof} We shall proceed in two steps.

\emph{Step 1. We assume that $P = \bigl( p(x,y) \bigr)_{x,y \in X}$ 
is stochastic and that $\rho(P)=1$.} 

Consider the matrix $Q$ of Lemma \ref{lem:q_str_substoch}. Let
$Q^n = \bigl( q^{(n)}(x,y) \bigr)_{x,y \in X}$ be its $n$-th matrix 
power. $q^{(n)}(x,y)$ is the probability that the Markov chain starting at $x$
is in $y$ at time $nk$ and does not make any forbidden sequence of transitions
in each of the discrete time intervals $[(j-1)k\,,\,jk]$ for 
$j \in \{1,\dots, n\}$. Therefore
$$
p^{(nk)}_F(x,y) \le q^{(n)}(x,y)\,,
$$
and also, by the same reasoning,  for $i=0, \dots, k-1$,
$$
p^{(nk+i)}_F(x,y) \le \sum_{z \in X} q^{(n)}(x,z)p^{(i)}_F(z,y)\,,\quad
i=0\, \dots, k-1.
$$
Therefore, for every $x \in X$ and $i=0, \dots, k-1$,
$$
\sum_{y \in X} p^{(nk+i)}_F(x,y) 
\le \sum_{z \in X} q^{(n)}(x,z) 
\underbrace{\sum_{y \in X}p^{(i)}_F(z,y)}_{\displaystyle\le 1}
\le (1-\eps_0)^n\,,
$$
since Lemma \ref{lem:q_str_substoch} implies that the row sums of the
matrix power $Q^n$ are bounded above by $(1-\eps_0)^n$.
We conclude that  
$$
\limsup_{n\to\infty} p_F^{(nk+i)}(x,y)^{1/(nk+i)} \leq (1-\eps_0)^{1/k}\,,
$$
so that $\rho_{x,y}(P_F) \le (1-\eps_0)^{1/k} = 1-\eps$,
where $\eps > 0$. 

\smallskip\noindent
\emph{Step 2. General case.}
We reduce this case to the previous one. 

Since $P$ is irreducible and
every row of $P$ has only finitely many non-zero entries,
Lemma \ref{lem:pruitt1} and Lemma \ref{lem:pruitt2}
guaranty the existence of a strictly
positive solution $h:X\to\R$ for the equation 
\begin{equation*}
P h = \rho(P) \cdot h,
\end{equation*}
that is, $h$ is \emph{$\rho(P)$-harmonic.}
Consider now the $h$-transform of the transition probabilities 
$p(e)$ of $P$, $e = (x,a,y) \in E$, given by
\begin{equation*}
p^h(e) = p^h(x,a,y) = \frac{p(x,a,y) h(y)}{\rho(P) h(x)}
\end{equation*}
and the associated transition matrix $P^h$ with entries 
$$
p^h(x,y) = \sum_{a\,:\, (x,a,y) \in E} p^h(x,a,y)\,.
$$
The Markov chain associated with $P^h$ is called the \emph{$h$-process.}

Then $\rho(P^h)=1$. Using uniform connectedness, we show that there is a 
constant $\bar \alpha > 0$ such that $p^h(e) \ge \bar\alpha$ for each
$e =(x,a,y)\in E$. Indeed, for such an edge, there is $k \le K$ such
that $d^+(y,x)=k$, whence
$$
\rho(P)^k h(y)  = \sum_{z \in X} p^{(k)}(y,z) h(z) \ge \alpha^k h(x)\,,
$$ 
so that
$$
p^h(x,a,y)\ge \bigl(\alpha/\rho(P)\bigr)^{k+1}\,.
$$ 
Recall that $K$ is the constant used in the definition of the uniform connectedness.
We can now choose $\bar \alpha = \bigl(\alpha/\rho(P)\bigr)^{K+1}$.
We see that with $P^h$ we are now in the situation of Step 1. Thus, forbidding
the transitions of $F$ for the Markov chain with transition matrix $P^h$, 
we get $\rho_{x,y}(P^h_F) \le 1 - \eps$ for all $x,y \in X$, where
$\eps > 0$.

We now show that $\rho_{x,y}(P^h_F) = \rho_{x,y}(P_F)/\rho(P)$, which will
conclude the proof.

For a path $\pi = e_1 \dots e_n$ from $x$ to $y$, let (as above) 
$\mathbb{P}(\pi)$ be the probability that the original Markov chain traverses
the edges of $\pi$ in $n$ successive steps, and let $\mathbb{P}^h(\pi)$
be the analogous probability with respect to the $h$-process.
Then
$$
\mathbb{P}^h(\pi) = \frac{\mathbb{P}(\pi)h(y)}{\rho(P)^nh(x)}\,.
$$
Let us write $\Pi_{x,y}^n(\neg F)$ for the set of all paths $\pi$ from
$x$ to $y$ with length $n$ for which $\ell(\pi)$ does not
contain a factor in $F$. Then the $n$-step transition probabilities
of the $h$-process with the transitions in $F$ forbidden are
$$
{p^h}^{(n)}_F(x,y) = \sum_{\pi \in \Pi_{x,y}^n(\neg F)} \mathbb{P}^h(\pi)
= \sum_{\pi \in \Pi_{x,y}^n(\neg F)}\frac{\mathbb{P}(\pi)h(y)}{\rho(P)^nh(x)}
= \frac{p^{(n)}_F(x,y)h(y)}{\rho(P)^nh(x)}
$$
Taking $n$-th roots and passing to the upper limit, 
we obtain the required
identity.
\end{proof}

With this result, it is now easy to deduce Theorem \ref{thm:A}.

\begin{proof}[Proof of Theorem \ref{thm:A}]
Since $(X,E,l)$ is deterministic with label alphabet $\Si$, 
the outdegree of every $x\in X$ is at most $|\Si|$. Equip the edges of $(X,E, \ell)$ 
with the transition probabilities 
$p(x,a,y) = 1/|\Si|$, when $(x,a,y)\in E$. Then the $n$-step transition 
probabilities of the resulting Markov chain are given by
\begin{equation*}
p^{(n)}(x,y)= \dfrac{\bigl|\{w \in L_{x,y}\,:\, \abs{w} = n\}\bigr|}
{|\Si|^n}.
\end{equation*}

Therefore, because $(X,E, \ell)$ is uniformly connected, we have
$$
\entr(X)=\entr(L_{x,y})  
=\limsup_{n\to\infty}\frac{1}{n} \log \bigl(p^n(x,y)|\Si|^n\bigr)
=\log \bigl(\rho(P)\cdot|\Si|\bigr).
$$
Analogously,
$$
\entr(L^F_{x,y})=\log\bigl(\rho_{x,y}(P_F)\cdot |\Si|\bigr).
$$
By Theorem \ref{thm:C}
$$
\sup_{x,y \in X} \rho_{x,y}(P_F) < \rho(P),
$$
and this implies that 
$$
\sup_{x,y\in X}\entr(L_{x,y}^F) < \entr(X)
$$
strictly.
\end{proof}

\textbf{Application to pairs of groups and their Schreier graphs}
\\[5pt]
Let $G$ be a finitely generated group and $K$ a (not necessary finitely 
generated) subgroup.  Let also $\Si$ be a finite alphabet and 
$\psi:\Si\rightarrow G$ be such that the set $\psi(\Si)$ generates $G$ as 
a semigroup. We extend $\psi$ to a monoid homomorphism from $\Si^*$
to $G$ by $\psi(w) = \psi(a_1)\cdots \psi(a_n)$, if $w = a_1 \dots a_n$
with $a_i \in \Si$ (and $\psi(\epsilon) = 1_G\,$). The mapping $\psi$ is
called a \emph{semigroup presentation} of $G$ in \cite{CeWo3}.  

The \emph{Schreier graph}  $X=X(G,K,\psi)$ has vertex set
\begin{equation*}
X=\{Kg: g\in G\}, 
\end{equation*}
the set of all right $K$-cosets in $G$, and the set of all labelled, 
directed edges $E$ is given by
\begin{equation*}
 E=\{e=(x,a,y): x=Kg, y=Kg \psi(a)\,,\; \text{where}\; g\in G\,,\;a\in\Si\}.
\end{equation*}
Note that the graph $X$ is fully deterministic and uniformly connected. 

The \emph{word problem} of $(G,K)$ with respect to $\psi$ is the language 
\begin{equation*}
 L(G,K,\psi)=\{w\in\Si^*: \psi(w)\in K\}.
\end{equation*}
The word problem for a recursively presented  group $G$ is the algorithmic problem 
of deciding whether two words represent the same element. Also, this
terminology is used in the context of formal language theory and goes back
at least to the seminal paper of {\sc Muller and Schupp}~\cite{muller_schupp_1983}.
For additional information, see also {\sc Muller and Schupp}~\cite{muller_schupp_1985}.
In their work, for a finitely generated group $G$ the \textbf{word problem}
$W(G)$ is the set of all words on the generators and their inverses which
represent the identity element of $G$.

If we consider the ``root'' vertex $o=K$ of the Schreier graph, then
in the notation of 
the introduction, 
we have $L(G,K,\psi)=L_{o,o}$, compare with
\cite[Lemma 2.4]{CeWo3}.

We can therefore apply Theorem \ref{thm:A} and Corollary \ref{cor:B} to the 
graph $X(G,K,\psi)$ in order to deduce that

\begin{cor}
The word problem of the pair $(G,K)$ with respect to any semigroup 
presentation $\psi$ is growth sensitive (with respect to forbidding an 
arbitrary non-empty finite subset $F\subset\Si^*$). 
\end{cor}

\end{document}